\documentclass[11pt]{article}

\usepackage{amsmath,amssymb,amsthm,amscd}
\usepackage[active]{srcltx}

\newcommand{\R}{\mathbb{R}}
\newcommand{\C}{\mathbb{C}}

\newcommand{\D}{{\cal D}}

\renewcommand{\S}{{\cal S}}
\newcommand{\B}{{\cal B}}
\newcommand{\K}{{\cal K}}
\newcommand{\F}{{\cal F}}
\newcommand{\pr}{{\rm pr}}
\newcommand{\la}{\lambda}

\newcommand{\h}{{\cal H}}

\newcommand{\Cyl}{{\rm Cyl}}
\newcommand{\id}{{\rm id}}	
\newcommand{\ot}{\otimes}

\newcommand{\bld}[1]{\boldsymbol{#1}}

\newcommand{\tp}{\tilde{p}}

\newcounter{mnotecount}[section]

\newtheorem{thr}{Theorem}
\newtheorem{lm}[thr]{Lemma}
\newtheorem{df}[thr]{Definition}

\newtheorem{pro}[thr]{Proposition}

\numberwithin{equation}{section}
\numberwithin{thr}{section}

\begin{document}

\title{Kinematic projective quantum states for Loop Quantum Gravity coupled to tensor fields\footnote{This is an author-created copyedited version of an article accepted for publication in Journal of Mathematical Physics. The definitive publisher authenticated version is available online at http://dx.doi.org/10.1063/1.4980014.}}
\author{Andrzej Oko{\l}\'ow}
\date{April 20, 2017}

\maketitle
\begin{center}
{\it  Institute of Theoretical Physics, Warsaw University\\ ul. Pasteura 5, 02-093 Warsaw, Poland\smallskip\\
oko@fuw.edu.pl}
\end{center}
\medskip

\begin{abstract}
We present a construction of kinematic quantum states for theories of tensor fields of an arbitrary sort. The construction is based on projective techniques by Kijowski. Applying projective quantum states for Loop Quantum Gravity obtained by Lan\'ery and Thiemann we construct quantum states for LQG coupled to tensor fields.   
\end{abstract}

\section{Introduction}

In the late 70's of the last century Kijowski \cite{kpt} introduced a method of constructing quantum states for field theories which is based on some projective techniques. Originally Kijowski applied his construction to theories of linear phase spaces like e.g. a scalar field theory. Few years ago this construction was generalized \cite{q-nonl,q-stat} to a certain class of theories of non-linear phase spaces like the Teleparallel Equivalent of General Relativity and recently Lan\'ery and Thiemann  generalized it even further \cite{proj-lt-I,proj-lt-II,proj-lt-III} in such a way that they were able to obtain by means of this method a new space of quantum states for Loop Quantum Gravity (LQG) \cite{proj-lqg-I}. Although nowadays the applicability of the Kijowski's method is quite broad it should be emphasized that this method neglects the dynamics of a theory and possible constraints on its phase space  and consequently the resulting quantum states are kinematic only. 

In this paper we will apply the Kijowski's method to construct kinematic quantum states for any tensor field theory and for LQG coupled to such a theory. A motivation for these constructions is the following.

Among important and interesting models of quantum gravity are those obtained by means of LQG methods applied to relational Dirac observables \cite{rov,ditt-prt,ditt} defined by a coupling of gravitational field and a matter field---this class of models encompasses both ones with reduced degrees of freedom (d.o.f.) applied in Loop Quantum Cosmology (see \cite{math-lqc} and reference therein) and ones with all d.o.f. \cite{kg-tt,grav-q}. The new space $\D_{\rm LQG}$ of quantum states for LQG constructed by Lan\'ery and Thiemann differs significantly from those used in the models presented in  \cite{kg-tt,grav-q}. Therefore it would be interesting to use the space $\D_{\rm LQG}$  to build  a model of LQG coupled to a matter field. 

To construct a space of quantum states for such a model of LQG we will restrict ourselves to tensor fields as matter fields and will construct by means of the Kijowski's method  a space of quantum states for an arbitrary tensor field theory---this space will be constructed in a background independent manner to fit the construction of $\D_{\rm LQG}$. Then we will have two field theories together with spaces of quantum states obtained for each theory by the Kijowski's method and will have to find a way to combine these two spaces into one which could serve for a theory defined by coupling the original two theories. We will find a general solution to this problem of combining two such spaces into one and will apply it to obtain the desired space of quantum states for LQG coupled to a tensor field theory.

This paper is organized as follows. Section 2 contains preliminaries. In Section 3 we will construct quantum states for an arbitrary tensor field theory, which in Section 4 will be combined with the space $\D_{\rm LQG}$. In Section 5 obtained results are shortly summarized.   

\section{Preliminaries}

\subsection{The projective construction of spaces of quantum states \label{out}}

Here we are going to outline the Kijowski's construction of quantum states---this outline is based on the newest formulation of the construction presented in \cite{mod-proj}. 

The point of departure for the construction is a phase space of a field theory which is usually an infinite dimensional space. One begins the construction by defining upon the phase space a family $\Lambda$ of finite physical systems---each such a system is obtained by a reduction of the infinite number of d.o.f. of the phase space to a finite one. The family $\Lambda$ should be defined in a very special way. Firstly, the finite systems constituting the family are supposed to represent altogether all d.o.f. of the original phase space. Secondly, it should be possible to equip the family with a directing relation $\geq$ such that $\la'\geq\la$ if $\la$ is a {\em subsystem} of $\la'$. Thirdly, one should be able to associate with every system $\la\in\Lambda$ a Hilbert space $\h_{\la}$ representing quantum states of the system in such a way that the family $\{\h_\la\}_{\la\in\Lambda}$ is extendable to a richer structure 
\begin{equation}
\Big(\Lambda,\h_\lambda,\tilde{\h}_{\lambda'\lambda},\Phi_{\lambda'\lambda},\Phi_{\lambda''\lambda'\lambda}\Big)
\label{quin}
\end{equation}
called a {\em  family of factorized Hilbert spaces} and defined as follows \cite{proj-lt-II}:
\begin{df}
A quintuplet \eqref{quin} is called a family of factorized Hilbert spaces if
\begin{enumerate}
\item $\Lambda$ is a directed set,
\item for every $\lambda\in\Lambda$ $\h_\lambda$ is a separable Hilbert space,
\item for every $\lambda'\geq \lambda$ $\tilde{\h}_{\lambda'\lambda}$ is a Hilbert space, and 
\begin{equation}
\Phi_{\lambda'\lambda}:\h_{\lambda'}\to\tilde{\h}_{\lambda'\lambda}\ot\h_\lambda
\label{Phi}
\end{equation}
is a Hilbert space isomorphism; moreover $\dim\tilde{\h}_{\lambda\lambda}=1$ and $\Phi_{\lambda\lambda}$ is trivial \footnote{Assume that $\h'$ is a one dimensional Hilbert space.  A Hilbert space isomorphism $\Phi:\h\to\h'\ot\h$ is trivial if there exists a normed element $e$ of $\h'$ such that $\Phi(\psi)=e\ot\psi$.}.      
\item for every $\lambda''\geq\lambda'\geq \lambda$ 
\[
\Phi_{\lambda''\lambda'\lambda}:\tilde{\h}_{\lambda''\lambda}\to\tilde{\h}_{\lambda''\lambda'}\ot\tilde{\h}_{\lambda'\lambda}
\]
is a Hilbert space isomorphism such that the following diagram 
\begin{equation}
\begin{CD}
\h_{\lambda''} @>\Phi_{\lambda''\lambda}>> \tilde{\h}_{\lambda''\lambda}\ot\h_\lambda\\
@VV\Phi_{\lambda''\lambda'}V       @VV\Phi_{\lambda'' \lambda'\lambda}\ot\id V \\
\tilde{\h}_{\lambda''\lambda'}\ot\h_{\lambda'}@>\id\ot\Phi_{\lambda'\lambda}>>  \tilde{\h}_{\lambda''\lambda'}\ot\tilde{\h}_{\lambda'\lambda}\ot\h_\lambda 
\end{CD}
\label{diagram}
\end{equation}
is commutative; moreover if $\lambda''=\lambda'$ or $\lambda'=\lambda$ then $\Phi_{\lambda''\lambda'\lambda}$ is trivial. 
\end{enumerate}
\label{ffHs}
\end{df}

Once a family of factorized Hilbert spaces is obtained one associates with every $\h_\la$ a $C^*$-algebra $\B_\la$ of all bounded operators on the Hilbert space. Then for every $\la'\geq\la$ the following map \cite{proj-lt-II} 
\[
\B_\lambda\ni a\mapsto \iota_{\lambda'\lambda}(a):=\Phi^{-1}_{\la'\la}\circ(\mathbf{1}_{\la'\la}\ot a)\circ\Phi_{\la'\la}\in\B_{\la'},
\]
where $\mathbf{1}_{\la'\la}$ is the identity operator on $\tilde{\h}_{\la'\la}$, is an injective $*$-homomorphism. For every triplet $\la''\geq\la'\geq\la$ the corresponding $*$-homomorphisms satisfy \cite{proj-lt-II}
\begin{equation}
\iota_{\la''\la'}\circ\iota_{\la'\la}=\iota_{\la''\la}
\label{iii}
\end{equation}
which means that a family $\{\B_\la,\iota_{\la'\la}\}_{\la\in\Lambda}$ is an {\em inductive family} of $C^*$-algebras.

Let $\S_\la$ be the set of all algebraic states on the $C^*$-algebra $\B_\la$. Then for every pair $\la'\geq \la$ one defines
\[
\S_{\la'}\ni\omega\mapsto\Pi_{\la\la'}(\omega):=\iota^*_{\la'\la}(\omega)\in \S_\la,
\]
where $\iota^*_{\la'\la}$ denotes a pull-back of states in $\S_{\la'}$ given by the $*$-homomorphism $\iota_{\la'\la}$. By virtue of \eqref{iii} for every triplet $\la''\geq\la'\geq\la$
\[
\Pi_{\la\la'}\circ\Pi_{\la'\la''}=\Pi_{\la\la''}
\]
and consequently, $\{\S_\la,\Pi_{\la\la'}\}_{\la\in\Lambda}$ is a {\em projective family}.

Finally, one defines a space $\S$ of quantum states for the field theory as {\em the projective limit} of the family $\{\S_\la,\Pi_{\la\la'}\}_{\la\in\Lambda}$---therefore the elements of $\S$ will be called {\em projective quantum states}. The limit $\S$ coincides with the set of all algebraic states on a $C^*$-algebra $\B$ defined as {\em the inductive limit} of the family $\{\B_\la,\iota_{\la'\la}\}_{\la\in\Lambda}$ \cite{mod-proj}. On the other hand, the algebra $\B$ is interpreted as an algebra of quantum observables of the theory \cite{kpt}.   

Each projective quantum state $s\in\S$ being an element of the projective limit is a special net $\{s_\la\}_{\la\in\Lambda}$ of states such that $s_\la\in\S_\la$. This means that the state $s$ contains information about quantum d.o.f. of every finite system $\la\in\Lambda$. It is clear that in reality we are not able to know exactly the state $s$, we are able to know exactly at most a state  $s_\la$ for some $\lambda$. Since $\la$ is a {\em finite} physical system information about its quantum states should be in principle available by means of an experiment. We can thus think of each system $\la$ as of a system constructed from d.o.f. which can be measured in an experiment. According to the original idea by Kijowski \cite{kpt} the state $s_\la$ should be treated as an experimentally available {\em approximation} of the state $s$. Consequently the whole Kijowski's construction can be seen as a method of defining a consistent family $\Lambda$ of experiments each of them provides us with an approximate knowledge about a state $s$ of a quantized field theory.           

Let us finally note that in the earlier papers \cite{kpt,q-nonl,q-stat}, \cite{proj-lt-II} and \cite{proj-lqg-I} a slightly different construction of  projective quantum states was used. That construction also requires a family of factorized Hilbert spaces to be obtained, but with every finite system ${\la}$ one does not associates  the space $\S_{\la}$ of all states on $\B_{\la}$ but one does a {\em proper} subset $\D_{{\la}}$ of $\S_{\la}$ which consists of all {\em normal states} on this algebra (or equivalently which consists of all {\em density operators} on the Hilbert space $\h_{{\la}}$). The spaces of normal states together with maps
\[
\{\pi_{\la\la'}:=\Pi_{\la\la'}\big|_{\D_{\la'}}\}
\]
form a projective family and in the previous papers spaces of projective quantum states were defined as projective limits of families of this sort. However, as noted in \cite{sl-phd} the construction using normal states may be flawed in some cases and therefore it was modified in \cite{mod-proj} to remove the possibility of the flaw and this modified construction is used in the present paper.

\subsection{A construction of a family of factorized Hilbert spaces \label{ffH-sec}}

Let us note that the task of constructing a space $\S$ of projective quantum states for a field theory reduces in fact to the task of constructing a family of factorized Hilbert spaces---this is because once such a family is obtained the further steps leading to the space $\S$ are straightforward. On the other hand, constructing a family of factorized Hilbert spaces for a field theory may be difficult and involved---see e.g. \cite{proj-lqg-I}. Here we will present shortly results obtained in \cite{q-nonl,mod-proj} which reduce the task of constructing a family of factorized Hilbert spaces to the task of constructing  a special directed set $\Lambda$ of finite physical systems. 

\paragraph{Phase space} Assume that a field theory possesses a phase space $P\times Q$, where $Q$ is a space of configurational variables and $P$ is a space of conjugate momenta---elements of $Q$ and elements of $P$ are to be understood as (collections of) fields on a manifold $\Sigma$ (usually $\Sigma$ is a three-dimensional manifold representing a spatial slice of a spacetime). 

\paragraph{Elementary d.o.f.} The first step of the construction of a family of factorized Hilbert spaces is the choice of so called {\em elementary degrees of freedom}: we say that a set ${\cal K}$ of real functions defined on $Q$ is a set of {\em configurational elementary d.o.f.} if the functions separate points of $Q$. An element of $\cal K$, that is, a configurational elementary d.o.f. will be usually denoted by $\kappa$  possibly with some indices. Similarly, we say that a set ${\cal F}$ of real functions defined on $P$ is a set of {\em momentum elementary d.o.f.} if the functions separate points of $P$. An element of $\cal F$, that is, a momentum elementary d.o.f. will be usually denoted by $\varphi$ possibly with some indices.    


\paragraph{Independent configurational d.o.f.} Let $K=\{\kappa_1,\ldots,\kappa_N\}$ be a finite set of configurational elementary d.o.f.. It defines on $Q$ an equivalence relation $\sim_K$: if $q,q'\in Q$ then $q\sim_K q'$ if for every $\kappa_\alpha\in K$ 
\[
\kappa_\alpha(q)=\kappa_\alpha(q').
\]         
Let $Q_K$ denote the space of equivalence classes of this relation:
\[
Q_K=Q\big/\!\sim_K
\] 
and let $\pr_K:Q\to Q_K$ be the canonical projection:
\[
\pr_K(q)=[q],
\]
where $[q]$ denotes the equivalence class of $q$. There exists a natural\footnote{The map $\tilde{K}$ is natural modulo the ordering of elements of $K$. Since any ordering is equally good for our purposes we will neglect this subtlety.} injective map $\tilde{K}:Q_K\to\R^N$ defined as follows:
\begin{equation}
Q_K\ni [q]\mapsto\tilde{K}([q]):=(\kappa_1(q),\ldots,\kappa_N(q))\in\R^N.
\label{k-inj}
\end{equation}
We will say that $K$ is a set of {\em independent} d.o.f. if the image of $\tilde{K}$ is an open subset of $\R^N$. 

If $K$ is a set of independent d.o.f. then the map $\tilde{K}$ can be used to pull-back the differential structure of the set $\tilde{K}(Q_K)$ onto $Q_K$. The differential manifold $Q_K$ will be called {\em reduced configuration space}. Note that the map $\tilde{K}$ defines a global coordinate frame $(x_\alpha)$ on the reduced configuration space $Q_K$:
\begin{equation}
Q_K\ni [q]\mapsto (x_1([q]),\ldots x_N([q])):=\tilde{K}([q])\in \R^N.
\label{lin-coor}
\end{equation}

We assume that the set $\cal K$ is chosen in such a way that if for two sets $K,K'$ of independent d.o.f. $Q_{K}=Q_{K'}$ then both maps $\tilde{K},\tilde{K}'$ define the same differential structure on $Q_{K}=Q_{K'}$.            

\paragraph{Cylindrical functions} Let $K$ be a set of independent d.o.f. and $Q_K$ a reduced configuration space. We say that a function $\Psi:Q\to \C$ is a {\em cylindrical function compatible with} $K$ if
\[
\Psi=\pr^*_K\psi
\]   
for a smooth function $\psi:Q_K\to\C$. We will denote by $\Cyl$ a complex linear space of functions on $Q$ spanned by all possible cylindrical functions.   

\paragraph{Momentum operators} We assume that every momentum d.o.f. $\varphi\in{\cal F}$ defines a ``{\em momentum operator}'' $\hat{\varphi}$ on $\Cyl$ by means of the Poisson bracket on $P\times Q$:
\[
\Cyl\ni\Psi\mapsto\hat{\varphi}\Psi:=\{\varphi,\Psi\}\in\Cyl
\]   
(elementary d.o.f. may be chosen in such a way that the Poisson bracket above is ill defined but even then it may be possible to obtain a well defined operator $\hat{\varphi}$ by means of a suitable regularization procedure---for such an example see \cite{q-nonl,acz}). Clearly, $\hat{\varphi}$ is a linear operator. Let $\hat{\cal F}$ be the real linear space spanned by all operators $\hat{\varphi}$:
\[
\hat{\cal F}={\rm span}_\R\{\ \hat{\varphi} \ | \ \varphi \in{\cal F} \  \}.
\]   
 
\paragraph{Directed set of finite systems} Let us denote by $\mathbf{K}$ the set of all sets of independent d.o.f. and by $\hat{\mathbf{F}}$ the set of all linear subspaces of $\hat{\cal F}$ of {\em finite dimension}. Suppose that $\Lambda$ is a subset of $\hat{\mathbf{F}}\times \mathbf{K}$ equipped with a relation $\geq$ such that $(\Lambda,\geq)$ is a directed set---this means that in this approach a finite physical system $\la\in\Lambda$ is just a pair $(\hat{F},K)$, where $K$ is a (finite) set of independent configurational d.o.f. (which defines the reduced configuration space $Q_K$) and $\hat{F}$ is a finite dimensional space of momentum operators. We assume that the elementary d.o.f. constituting the spaces ${\cal K}$ and $\cal F$ and the directed set $(\Lambda,\geq)$ are chosen in such a way that the following {\em Conditions} hold \cite{q-nonl}:
\begin{enumerate}
\item 
\begin{enumerate}
\item for each finite set $K_0$ of configurational elementary d.o.f.  there exists $(\hat{F},K)\in\Lambda$ such that each $\kappa\in K_0$ is a cylindrical function compatible with $K$; \label{k-Lambda}
\item for each finite set $F_0$ of momentum elementary d.o.f. there exists $(\hat{F},K)\in\Lambda$ such that $\hat{\varphi}\in\hat{F}$ for every $\varphi\in F_0$; \label{f-Lambda}
\end{enumerate}
\item \label{RN} 
if $(\hat{F},K)\in\Lambda$ then the image of the map $\tilde{K}$ given by \eqref{k-inj} is $\R^N$, where $N$ is the number of elements of $K$---in other words, $\tilde{K}$ is a bijection and consequently  
\[
Q_K\cong\R^N.
\] 
\item 
if $(\hat{F},K)\in\Lambda$, then 
\begin{enumerate}
\item for every $\hat{\varphi}\in \hat{\F}$ and for every cylindrical function $\Psi=\pr_K^*\psi$ compatible with $K=\{\kappa_1,\ldots,\kappa_N\}$ 
\[
\hat{\varphi}\Psi=\sum_{\alpha=1}^N\Big(\pr^*_K\partial_{x_\alpha}\psi\Big)\hat{\varphi}\kappa_\alpha,
\]   
where $\{\partial_{x_\alpha}\}$ are vector fields on $Q_K$ given by the global coordinate frame \eqref{lin-coor}; \label{comp-f} 
\item for every $\hat{\varphi}\in \hat{\F}$ and for every $\kappa\in K$ the cylindrical function $\hat{\varphi}\kappa$ is a real {\em constant} function on $Q$; \label{const}
\end{enumerate}
\item if $(\hat{F},K)\in\Lambda$ and $K=\{\kappa_{1},\ldots,\kappa_{N}\}$ then $\dim\hat{F}=N$; moreover, if $(\hat{\varphi}_1,\ldots,\hat{\varphi}_N)$ is a basis of $\hat{F}$ then an $N\times N$ matrix $G=(G_{\beta\alpha})$ of components
\begin{equation}
G_{\beta\alpha}:=\hat{\varphi}_\beta\kappa_\alpha
\label{Gij}
\end{equation}
is {\em non-degenerate}. \label{non-deg}
\item  if $(\hat{F},K'),(\hat{F},K)\in\Lambda$ and $Q_{K'}=Q_{K}$ then  $(\hat{F},K')\geq(\hat{F},K)$; \label{Q'=Q} 
\item if $(\hat{F}',K')\geq(\hat{F},K)$ then 
\begin{enumerate}
\item each d.o.f. in $K$ is {\em a linear combination} of d.o.f. in $K'$; \label{lin-comb}
\item $\hat{F}$ is {\em a linear subspace} of $\hat{F}'$. \label{FF'}
\end{enumerate} 
\end{enumerate} 
        
It was shown in \cite{mod-proj} that if these conditions are satisfied then there exists a family \eqref{quin} of factorized Hilbert spaces associated with the directed set $(\Lambda,\geq)$---this family is a natural extension of a family $\{\h_\la\}_{\la\in\Lambda}$ of Hilbert spaces defined as follows \cite{q-nonl}: given $\lambda=(\hat{F},K)\in \Lambda$ 
\[
\h_\lambda:=L^2(Q_K,d\mu_\lambda),
\] 
where $d\mu_\lambda$ is a Lebesgue measure on $Q_K\cong\R^N$ defined by the coordinates \eqref{lin-coor}. 


\subsection{A useful proposition \label{u-prop}}

Above we described the results of \cite{q-nonl,mod-proj} which reduce the task of constructing a family of factorized Hilbert spaces to the one of constructing the special directed set $(\Lambda,\geq)$. Besides that in \cite{q-nonl} we proved some quite general facts which can be used while constructing such a set for a phase space. Below we will present another proposition of this kind, which guarantees that a subset of $\hat{\mathbf{F}}\times \mathbf{K}$ built according to a pattern is naturally a directed set---in fact, this proposition is a generalization of two particular lemmas to be found in \cite{q-nonl} and \cite{q-stat}.   

Let $(\Gamma,\geq)$ be a directed set and let $\mathbf{K}_\Gamma$ be a subset of $\mathbf{K}$ elements of which are labeled by elements of $\Gamma$ i.e.
\[
\mathbf{K}_\Gamma=\{K_\gamma\}_{\gamma\in\Gamma},
\]   
where each $K_\gamma$ is a set of independent d.o.f.. 

We say that a pair $(\hat{F},K)\in\hat{\mathbf{F}}\times \mathbf{K}$, where $K=\{\kappa_1,\ldots,\kappa_N\}$ is {\em non-degenerate} if 
\begin{enumerate}
\item for every $\hat{\varphi}\in \hat{F}$ and every $\kappa_\alpha\in K$ the function $\hat{\varphi}\kappa_\alpha$ is a real {\em constant} function, 
\item the dimension of $\hat{F}$ is equal $N$,
\item for a basis $(\hat{\varphi}_1,\ldots,\hat{\varphi}_N)$ of $\hat{F}$ the matrix \eqref{Gij} is non-degenerate.
\end{enumerate}

Consider now a set $\Lambda$ consisting of all non-degenerate pairs $(\hat{F},K_{\gamma})\in \hat{\mathbf{F}}\times \mathbf{K}_\Gamma$, where $\gamma$ runs through $\Gamma$.  On $\Lambda$ there exists a natural relation $\geq$ defined as follows: we say that $\lambda'=(\hat{F}',K_{\gamma'})$ is greater or equal to $\lambda=(\hat{F},K_{\gamma})$,
\[
\lambda'\geq\lambda,
\]   
if 
\begin{align*}
\hat{F}'&\supset\hat{F},& \gamma'&\geq\gamma.
\end{align*}

\begin{pro}
Suppose that 
\begin{enumerate}
\item \label{gamma>N} for every $\gamma\in\Gamma$ and for every natural number $N$ there exists $\gamma'\in \Gamma$ such that $\gamma'\geq\gamma$ and the number of elements of $K_{\gamma'}$ is greater than $N$, 
\item \label{lin-comb-Kg} if $\gamma'\geq\gamma$ then each d.o.f. in $K_{\gamma}$ is a cylindrical function compatible with  $K_{\gamma'}$,
\item \label{F-Kg-ndeg} for every $\gamma\in\Gamma$ there exists $\hat{F}\in\mathbf{F}$ such that $(\hat{F},K_{\gamma})\in \Lambda$.
\end{enumerate}
Assume moreover that the set $(\Lambda,\geq)$ satisfies Conditions \ref{k-Lambda}, \ref{comp-f} and \ref{const} (presented in Section \ref{ffH-sec}). Then $(\Lambda,\geq)$ is a directed set.
\label{Lambda-dir}     
\end{pro}

To prove this proposition we will use some facts proven in \cite{q-nonl}. Consider a set $\bld{\Psi}\subset \Cyl$ and a finite set $\{\hat{\varphi}_1,\ldots,\hat{\varphi}_N\}\subset\hat{\F}$. We will say that the operators  $\{\hat{\varphi}_1,\ldots,\hat{\varphi}_N\}$ are {\em linearly independent on} $\bld{\Psi}$ if the operators restricted to $\bld{\Psi}$ are linearly independent. Now let us state the facts:

\begin{lm}
Let $\Cyl_K$ be the set of all cylindrical functions compatible with a set $K$ of independent d.o.f..  Assume that operators $\{\hat{\varphi}_1,\ldots,\hat{\varphi}_M\}\subset \hat{\F}$ act on elements of $\Cyl_K$ according to the formula in Condition \ref{comp-f} (Section \ref{ffH-sec}). If $\{\hat{\varphi}_1,\ldots,\hat{\varphi}_M\}\subset \hat{\F}$ are linearly independent on a subset $\bld{\Psi}$ of $\Cyl_K$ then they are linearly independent on $K$.       
\label{cyl-K}
\end{lm}

\begin{pro}
Let $\Lambda$ be a subset of $\hat{\mathbf{F}}\times\mathbf{K}$ which satisfies Conditions \ref{k-Lambda} and \ref{comp-f} (Section \ref{ffH-sec}). Then for every finite set $\{\hat{\varphi}_1,\ldots,\hat{\varphi}_M\}\subset\hat{\cal F}$ of linearly independent operators there exists a pair $(\hat{F},K)\in\Lambda$ such that the operators are linearly independent on $K$.
\label{Lambda-pr}   
\end{pro} 
\noindent To be able to use these facts in the proof of Proposition \ref{Lambda-dir} we assumed in the proposition that the set $(\Lambda,\geq)$ satisfies Conditions \ref{k-Lambda} and \ref{comp-f}. 

\begin{proof}[Proof of Proposition \ref{Lambda-dir}]

Clearly, the relation $\geq$ is transitive. Therefore we have to prove  only that for any $\lambda',\lambda\in\Lambda$ there exists $\lambda''\in\Lambda$ such that $\lambda''\geq\lambda',\lambda$. 

Let us fix $\lambda'=(\hat{F}',K_{{\gamma}'})$ and $\lambda=(\hat{F},K_{{\gamma}})$. We define $\hat{F}_0$ as a linear subspace of $\hat{\F}$ spanned by elements of $\hat{F}'\cup\hat{F}$ and choose a basis $(\hat{\varphi}_1,\ldots,\hat{\varphi}_M)$ of $\hat{F}_0$. Proposition \ref{Lambda-pr} and the definition of the set $\Lambda$ guarantee that there exists an element ${\gamma}_0\in\Gamma$ such that the operators $(\hat{\varphi}_1,\ldots,\hat{\varphi}_M)$ remain  linearly independent when restricted to $K_{{\gamma}_0}$. By virtue of Assumption \ref{gamma>N} of Proposition \ref{Lambda-dir} it is possible to choose an element  ${\gamma}''\in\Gamma$ such that $(i)$ ${\gamma}''\geq {\gamma}_0,{\gamma}',{\gamma}$ and $(ii)$ the number $N$ of elements of $K_{{\gamma}''}$ is greater than $\dim \hat{F}_0=M$. Since $\gamma''\geq \gamma_0$ each d.o.f. in $K_{\gamma_0}$ is a cylindrical function compatible with $K_{\gamma''}$ (this is ensured by Assumption \ref{lin-comb-Kg}) and according to Lemma \ref{cyl-K} the operators $(\hat{\varphi}_1,\ldots,\hat{\varphi}_M)$ are linearly independent on $K_{{\gamma}''}$.     

Consider now a matrix $G^0$ of components
\[
G^0_{\beta\alpha}:=\hat{\varphi}_{\beta}\kappa_{\alpha},
\]          
where $\{\kappa_1,\ldots,\kappa_{N}\}=K_{{\gamma}''}$---it follows from Condition \ref{const} that the components can be treated as real numbers. Clearly, $G^0$ is a matrix of $N$ columns and $M$ rows. Since the operators $(\hat{\varphi}_1,\ldots,\hat{\varphi}_M)$ are linearly independent on $K_{{\gamma}''}$ the rank of $G^0$ is equal $M<N$. By means of the following operations: $(i)$ multiplying a row of by a non-zero real number, $(ii)$ adding to a row a linear combination of other rows $(iii)$ reordering the rows and $(iv)$ reordering the columns one can obtain from $G^0$ a matrix $G^1$ of the following form
\[
G^1=
\begin{pmatrix}
\mathbf{1} & G' 
\end{pmatrix},
\]             
where $\mathbf{1}$ denotes $M\times M$ unit matrix and $G'$ does an $M\times(N-M)$ matrix. Note that the operations $(i)$--$(iii)$ used to transform $G^0$ to $G^1$  correspond to a transformation of the basis $(\hat{\varphi}_1,\ldots,\hat{\varphi}_M)$ to a basis $(\hat{\varphi}'_1,\ldots,\hat{\varphi}'_M)$ of $\hat{F}_0$, while the operation $(iv)$ corresponds to renumbering the d.o.f. in $K_{{\gamma}''}$: $\kappa_{\alpha}\mapsto \kappa'_{\alpha}:=\kappa_{\sigma(\alpha)}$, where $\sigma$ denotes a permutation of the sequence $(1,\ldots,N)$. Thus  
\[
G^1_{\beta\alpha}=\hat{\varphi}'_{\beta}\kappa'_{\alpha}.
\]

By virtue of Assumption \ref{F-Kg-ndeg} there exists $\hat{F}^0\in\hat{\mathbf{F}}$ such that $(\hat{F}^0,K_{\gamma''})\in\Lambda$ which means that the pair $(\hat{F}^0,K_{\gamma''})$ is non-degenerate. Therefore one can choose a basis $(\hat{\varphi}^0_1,\ldots,\hat{\varphi}^0_{N})$ of $\hat{F}^0$ such that 
\[
\hat{\varphi}^0_{\beta}\kappa'_{\alpha}=\delta_{\beta\alpha},
\]
where $\{\kappa'_1,\ldots,\kappa'_N\}=K_{\gamma''}$. Thus if
\[
\Big(\hat{\varphi}''_1,\ldots,\hat{\varphi}''_{N}\Big):=\Big(\hat{\varphi}'_1,\ldots,\hat{\varphi}'_{M},\hat{\varphi}^0_{M+1},\ldots,\hat{\varphi}^0_{N}\Big)
\]
then a $N\times N$ matrix $G=(G_{\beta\alpha})$ of components
\[
G_{\beta\alpha}:=\hat{\varphi}''_{\beta}\kappa'_{\alpha}
\] 
reads
\[
G=
\begin{pmatrix}
\mathbf{1} & G'\\
\mathbf{0} & \mathbf{1}' 
\end{pmatrix}
\]  
---here $\mathbf{0}$ denotes a zero $(N-M)\times M$ matrix, and $\mathbf{1}'$ does a unit $(N-M)\times(N-M)$ matrix. It is obvious that the matrix $G$ is non-degenerate which in particular allows us to conclude that the operators $(\hat{\varphi}''_1,\ldots,\hat{\varphi}''_{N})$ are linearly independent.

Let
\[
\hat{F}'':={\rm span}_{\mathbb{R}} \{\hat{\varphi}''_1,\ldots,\hat{\varphi}''_{N}\}\in\hat{\mathbf{F}}
\]
Then $\lambda'':=(\hat{F}'',K_{{\gamma}''})$ is the desired element of $\Lambda$ such that $\lambda''\geq\lambda',\lambda$.        

\end{proof}

\section{Projective quantum states for a theory of tensor fields \label{sec-tf}}

Using the method described in Section \ref{ffH-sec} we will construct below in a {\em background independent} manner a directed set $(\Lambda,\geq)$ for a theory of tensor fields of a particular sort---canonical variables of this theory will be symmetric tensor fields of type $\binom{0}{2}$ on a three dimensional manifold and corresponding conjugate momenta. The set $(\Lambda,\geq)$ obtained here will satisfy all Conditions listed in Section \ref{ffH-sec} which means that together with the set we will obtain a family of factorized Hilbert spaces and a space $\S$ of projective quantum states for the theory. 

Then we will show how to generalize this particular construction of the space $\S$ to a theory of tensor fields of an arbitrary sort.
 
\subsection{Projective quantum states for a theory of symmetric tensor fields of type $\binom{0}{2}$}

\paragraph{Phase space} Let $Q$ be the set of all symmetric tensor fields of type $\binom{0}{2}$ on a three dimensional manifold $\Sigma$, and $P$ the set of all symmetric  tensor densities of type $\binom{2}{0}$ and weight $1$. In a (local) coordinate system $(x^i)$ on $\Sigma$ elements $q\in Q$ and $\tilde{p}\in P$ are represented by components  
\begin{align*}
q_{ij}&=q_{ji},&\tp^{ij}&=\tp^{ji}.
\end{align*}

Let $F$ be a functional on $P\times Q$.  A functional derivative $\delta F/\delta g_{ij}$ is a {\em symmetric} tensor density of type $\binom{2}{0}$ and weight $1$ and a functional derivative $\delta F/\delta \tilde{p}^{ij}$ is a {\em symmetric} tensor field of type $\binom{0}{2}$ such that 
\[
\delta F=\int_{\Sigma} \frac{\delta F}{\delta q_{ij}}\delta q_{ij}+\frac{\delta F}{\delta \tp^{ij} }\delta \tilde{p}^{ij}.
\]

We define a Poisson bracket between two functionals $F,G$ on $P\times Q$ as follows:
\begin{equation}
\{F,G\}:=\int_{\Sigma}\frac{\delta F}{\delta q_{ij}}\frac{\delta G}{\delta \tp^{ij}}-\frac{\delta G}{\delta q_{ij}}\frac{\delta F}{\delta \tp^{ij} },
\label{poiss}
\end{equation}
which means that the variable $p$ is the momentum conjugate to $q$. The product $P\times Q$ is thus a phase space.    

\paragraph{Elementary d.o.f.} Perhaps the most simple and the most natural way to produce a real number from the variable $q$ is to evaluate it on a pair of tangent vectors. Thus for  $Y,Y'\in T_y\Sigma$ we define  
\[
Q\ni q\mapsto \kappa_{YY'}(q):=q(Y,Y')\in \R
\]
and will treat this function as a configurational elementary d.o.f.. We will say that $\kappa_{YY'}$ is supported at $y$. Obviously,
\[
\kappa_{YY'}=\kappa_{Y'Y}
\]
We denote by $\cal K$ the set of all configurational  d.o.f.:
\[
{\cal K}=\{\ \kappa_{YY'} \ | \ Y,Y'\in T_y\Sigma,\ y\in\Sigma \ \}.
\]  

Let $\omega,\omega'$ be one-forms on $\Sigma$ of compact support. Then 
\[
\tp(\omega,\omega'):=\tp^{ij}\omega_i\omega'_j
\]  
is a scalar density on $\Sigma$ which can be naturally integrated over the manifold:
\begin{equation}
P\ni \tp\mapsto \varphi_{\omega\omega'}(\tp):=\int_\Sigma\tp(\omega,\omega')\in \R.
\label{phi-omom}
\end{equation}
If the function $\varphi_{\omega\omega'}$ is non-zero then we will treat it as a momentum elementary d.o.f.. Clearly,
\[
\varphi_{\omega\omega'}=\varphi_{\omega'\omega}.
\]
The set of all momentum elementary d.o.f. given by all one-forms $\{\omega,\omega'\}$ on $\Sigma$ of compact support will be denoted by $\cal F$.   

Let us emphasize that in the definition of $\K$ we used all possible pairs of tangent vectors and in the definition of ${\cal F}$ we used all possible pairs of one-forms which give non-zero functions \eqref{phi-omom}. Therefore no pair of vectors and no pair of one-forms is distinguished by the definitions which is a necessary condition for the construction of projective quantum states to be background independent.   

\paragraph{Discrete frames} Let $N>0$ be a natural number. A discrete frame $\gamma$ is a  collection of $3N$ vectors tangent to $\Sigma$---   
\[
\gamma=\bigcup_{I=1}^N\{ \ e_{I1},e_{I2},e_{I3}\ \}
\]  
such that 
\begin{enumerate}
\item for each $I\in\{1,\ldots,N\}$ the triplet $(e_{I1},e_{I2},e_{I3})$ is a basis of $T_{y_I}\Sigma$,
\item the points $\{y_1,\ldots,y_N\}$ are pairwise distinct.
\end{enumerate}
The points $\{y_1,\ldots,y_N\}$ will be called {\em points underlying the frame} $\gamma$ and the number of the points will be denoted by $N_\gamma$. 
    
The set $\Gamma$ of all discrete frames in $\Sigma$ is a directed set with a relation $\geq$ defined as follows: $\gamma'\geq \gamma$ if the set of points underlying $\gamma$ is a subset of the set of points underlying $\gamma'$.        

\paragraph{Finite sets of configurational d.o.f.} Let $\gamma$ be a discrete frame. Each pair of vectors $Y,Y'\in \gamma$ tangent to $\Sigma$ at the same point defines a d.o.f. $\kappa_{YY'}$. The set of all such d.o.f. will be denoted by $K_\gamma$. Clearly, the set $K_\gamma$ consists of $6N_\gamma$ elements. There holds the following obvious lemma: 
\begin{lm}
Let $K_\gamma=\{\kappa_1,\ldots,\kappa_{6N}\}$. For every $(x_1,\ldots,x_{6N})\in \R^{6N}$ there exists $q\in Q$ such that
\[
\kappa_\alpha(q)=x_\alpha.
\] 
\label{Kg-R6N}
\end{lm}     
\noindent This means that the map $\tilde{K}_\gamma$ is a bijection, 
\begin{equation}
Q_{K_\gamma}\cong \R^{6N_\gamma},
\label{QKg-R6N}
\end{equation}
$K_\gamma$ is a set of independent d.o.f. and $Q_{K_\gamma}$ is a reduced configuration space.

\begin{lm}
For every finite set $K$ of configurational d.o.f. there exists a discrete frame $\gamma$ such that each d.o.f. in $K$ is a linear combination of d.o.f. in $K_\gamma$.   
\label{K-linKg}
\end{lm}

\begin{proof}
There exists a finite set of points $\{y_1,\ldots,y_N \}$ of $\Sigma$ such that each $\kappa\in K$ is defined by a pair of vectors tangent to $\Sigma$ at one of these points. Let $\gamma$ be any discrete frame such that its underlying points are $\{y_1,\ldots,y_N \}$. Let $\kappa_{YY'}\in K$, where $Y,Y'\in T_{y_I}\Sigma$ and let $(e_{1},e_{2},e_{3})$ be the basis of $T_{y_I}\Sigma$ belonging to $\gamma$. Then both $Y$ and $Y'$ are linear combinations of the vectors $(e_{1},e_{2},e_{3})$ and consequently $\kappa_{YY'}$ is a linear combination of d.o.f. $\{\kappa_{e_{i}e_{j}}\}$.              
\end{proof}

\begin{lm}
$\gamma'\geq\gamma$ if and only if every d.o.f. in $K_\gamma$ is a linear combination of d.o.f. in $K_{\gamma'}$.  
\label{lm-lin-comb}
\end{lm}
\begin{proof}
Assume that $\gamma'\geq\gamma$. Let $\kappa_{YY'}\in K_\gamma$. Then $Y,Y'\in T_y\Sigma$ and $y$ is a point underlying the frame $\gamma'$. Let $(e'_1,e'_2,e'_3)$ be a basis of $T_y\Sigma$ belonging to $\gamma'$. Then both $Y$ and $Y'$ is a linear combination of $(e'_1,e'_2,e'_3)$ and consequently, $\kappa_{YY'}$ is a linear combination of d.o.f. in $K_{\gamma'}$ defined by the basis.               

Assume now that every d.o.f. in $K_\gamma$ is a linear combination of d.o.f. in $K_{\gamma'}$. Thus if $\kappa_{YY'}\in K_\gamma$ and $K_{\gamma'}=\{\kappa'_\alpha\}$ then
\[
\kappa_{YY'}=A^\alpha\kappa'_\alpha.
\]
Suppose that $Y,Y'\in T_y\Sigma$. Then if no $\kappa'_\alpha\in K_{\gamma'}$ is supported at $y$ then all the coefficients $\{A^\alpha\}$ must be equal to zero. But this cannot be the case since $\kappa_{YY'}$ is a non-zero function. Thus $y$ is one of the points underlying $\gamma'$. This means that  $\gamma'\geq \gamma$.           
\end{proof}

Let $\mathbf{K}$ be the set of all finite sets of independent d.o.f., and let $\mathbf{K}_\Gamma$ be its subset given by all discrete frames:
\[
\mathbf{K}_\Gamma:=\{ \ K_\gamma\ | \ \gamma\in\Gamma\ \}.
\]
By virtue of Lemmas \ref{Kg-R6N} and \ref{K-linKg} the set $\mathbf{K}_\Gamma$ satisfies the requirements imposed on a set $\mathbf{K}'$ by the following proposition \cite{q-nonl}:    

\begin{pro}
Suppose that there exists a subset $\mathbf{K}'$ of $\mathbf{K}$ such that  for every finite set $K_0$ of configurational elementary d.o.f. there exists $K'_0\in\mathbf{K}'$ satisfying the following conditions: 
\begin{enumerate}
\item the map $\tilde{K}'_0$ is a bijection; 
\item each d.o.f. in $K_0$ is a linear combination of d.o.f. in $K'_0$.  
\end{enumerate} 
Then
\begin{enumerate}
\item for every set $K\in\mathbf{K}$ the map $\tilde{K}$ is a bijection. Consequently, $Q_K\cong \R^N$ with $N$ being the number of elements of $K$ and the map $\tilde{K}$ defines a linear structure on $Q_K$ being the pull-back of the linear structure on $\R^N$; if $Q_{K}=Q_{K'}$ for some other set $K'\in\mathbf{K}$ then the linear structures defined on the space by $\tilde{K}$ and $\tilde{K}'$ coincide.
\item for every element $\Psi\in\Cyl$ there exists a set $K\in\mathbf{K}'$ such that $\Psi$ is compatible with $K$.    
\end{enumerate}
\label{big-pro}
\end{pro}
\noindent The first assertion of the proposition guarantees in particular that if $Q_K=Q_{K'}$ then the differential structure on $Q_K$ coincides with that on $Q_{K'}$. Therefore the space $\Cyl$ spanned by cylindrical functions on $Q$ is well defined. The second assertion means that every element of $\Cyl$ is a cylindrical function compatible with some $K_\gamma\in\mathbf{K}_\Gamma$.       

\paragraph{Momentum operators} Each $\varphi_{\omega\omega'}\in{\cal F}$ defines an operator on $\Cyl$ as follows: 
\[
\Cyl\ni\Psi\mapsto\hat{\varphi}_{\omega\omega'}\Psi:=\{\varphi_{\omega\omega'},\Psi\},
\]   
where the r.h.s. is a function on $Q$. Let $\hat{\cal F}$ denote a real linear space spanned by all such operators:
\[
\hat{\cal F}={\rm span}_\R\{\ \hat{\varphi}_{\omega\omega'}\ | \ \varphi_{\omega\omega'}\in{\cal F} \  \}.
\]   
We know already that if $\Psi\in\Cyl$ then $\Psi=\pr^*_{K_\gamma}\psi$ for some $K_\gamma\in\mathbf{K}_\Gamma$. Let $(x_\alpha)$ be the natural coordinate system on $Q_{K_\gamma}$ given by the map $\tilde{K}_\gamma$ (see Equation \eqref{lin-coor}). Then for any $\hat{\varphi}\in\hat{\cal F}$  
\[
(\hat{\varphi}\Psi)(q)=\sum_\alpha\frac{\partial\psi}{\partial x_\alpha}(\kappa_\beta(q))(\hat{\varphi}\kappa_\alpha)(q)=\Big(\sum_\alpha\big(\pr^*_{K_\gamma}\partial_{x^\alpha}\psi\big)\hat{\varphi}\kappa_\alpha\Big)(q)
\]   
hence
\begin{equation}
\hat{\varphi}\Psi=\sum_\alpha\big(\pr^*_{K_\gamma}\partial_{x^\alpha}\psi\big)\hat{\varphi}\kappa_\alpha
\label{phi-Psi}
\end{equation}
Using \eqref{poiss} it is easy to check that
\begin{equation}
\hat{\varphi}_{\omega\omega'}\kappa_{YY'}=-\frac{1}{2}(\omega(Y)\omega'(Y')+\omega(Y')\omega'(Y))
\label{phi-omom-kyy}
\end{equation}
which is a real constant function on $Q$. This means that for every $\hat{\varphi}\in\hat{\cal F}$ and for every d.o.f. $\kappa$ the function $\hat{\varphi}\kappa$ is a {\em real constant} function on $Q$. Taking into account the formula \eqref{phi-Psi} we see that each $\hat{\varphi}\Psi$ is again a cylindrical function compatible with the same $K_\gamma$ as $\Psi$ is. Thus all the operators in $\hat{\cal F}$ preserve the space $\Cyl$.    

\begin{lm}
Let $\gamma$ be a discrete frame and $K_\gamma=\{\kappa_1,\ldots,\kappa_{N}\}$. Then there exists a set $(\hat{\varphi}_1,\ldots,\hat{\varphi}_N)\subset \hat{\cal F}$ such that
\[
\hat{\varphi}_\beta\kappa_\alpha=\delta_{\beta\alpha}.
\] 
\label{phi-k-delta}
\end{lm}
\begin{proof}
Let $\{y_1,\ldots,y_N\}$ be a set of points underlying the frame $\gamma$. Given point $y_I$ in the set it is possible to choose three one-forms $(\omega^1_I,\omega^2_I,\omega^3_I)$ such that $(i)$ for the basis $(e_{I1},e_{I2},e_{I3})$ of $T_{y_I}\Sigma$ belonging to $\gamma$
\[
\omega^i_I(e_{Ij})=\delta^i_j
\]    
and $(ii)$ $y_I$ is the only underlying point of $\gamma$ which belongs to the supports of all $(\omega^1_I,\omega^2_I,\omega^3_I)$. Then by virtue of \eqref{phi-omom-kyy}
\[
\hat{\varphi}_{\omega^i_I\omega^j_I}\kappa_{e_{Jk}e_{Jl}}=-\frac{1}{2}\big(\omega^i_I(e_{Jk})\omega^j_I(e_{Jl})+\omega^i_I(e_{Jl})\omega^j_I(e_{Jk})\big)=-\frac{1}{2}\delta_{IJ}(\delta^i{}_k\delta^j{}_l+\delta^i{}_l\delta^j{}_k).
\]   
Rescaling appropriately the operators $\{\hat{\varphi}_{\omega^i_I\omega^j_I}\}$ one obtains the desired operators $\{\hat{\varphi}_\beta\}$.    
\end{proof}

\paragraph{A directed set $(\Lambda,\geq)$} Recall that in Section \ref{u-prop} we denoted by $\hat{\mathbf{F}}$ the set of all finite dimensional linear subspaces of $\hat{\cal F}$. Using pairs $(\hat{F},K_\gamma)\in\hat{\mathbf{F}}\times \mathbf{K}_\Gamma$ we define a set $\Lambda$ and a relation $\geq$ on it exactly as it was done in Section \ref{u-prop}.  

This set $(\Lambda,\geq)$ is a directed set. To justify this statement it is enough to show that the sets $\mathbf{K}_\Gamma$ and $\Lambda$ defined in the present section satisfy all assumptions of Proposition \ref{Lambda-dir}. It is clear that the set $\mathbf{K}_\Gamma$ satisfies Assumption \ref{gamma>N} of the proposition. This set meets Assumption \ref{lin-comb-Kg} by virtue of Lemma \ref{lm-lin-comb}.   On the other hand Lemma \ref{phi-k-delta} guarantees that Assumption \ref{F-Kg-ndeg} is satisfied. Moreover, it follows from Lemmas \ref{K-linKg} and \ref{phi-k-delta} that the set $(\Lambda,\geq)$ satisfies Condition \ref{k-Lambda} (Section \ref{ffH-sec}), and Equation \eqref{phi-Psi} means that the set meets Condition \ref{comp-f}. We already concluded that for every $\hat{\varphi}\in\hat{\cal F}$ and $\kappa\in{\cal K}$ the function $\hat{\varphi}\kappa$ is real and constant as it is required by Condition \ref{const}.

Now let us show that the directed set $(\Lambda,\geq)$ just constructed satisfies all remaining Conditions listed in Section \ref{ffH-sec}. 

Let us start with Condition \ref{f-Lambda}. Consider a set $F_0=\{\varphi_{\omega_1\omega'_1},\ldots,\varphi_{\omega_N\omega'_N}\}\subset{\cal F}$. Let us fix $I\in\{1,\ldots,N\}$ and consider the momentum d.o.f. $\varphi_{\omega_I\omega'_I}\in F_0$. Since this d.o.f. is non-zero there exists a discrete frame $\gamma_I=\{e_1,e_2,e_3\}$ consisting of a basis of some $T_y\Sigma$ such that the following sextuplet
\[
\{\hat{\varphi}_{\omega_I\omega'_I}\kappa_{e_ie_j}\}_{i\leq j}
\]
contains at least one non-zero number. Lemma \ref{phi-k-delta} guarantees that there exists operators $\{\hat{\varphi}_1,\ldots,\hat{\varphi}_5\}$ such that
\[
\hat{F}_I:={\rm span}_{\R}\{\hat{\varphi}_{\omega_I\omega'_I},\hat{\varphi}_1,\ldots,\hat{\varphi}_5\}\in\hat{\mathbf{F}}
\]
and $K_{\gamma_I}$ form a non-degenerate pair $\lambda_I=(\hat{F}_I,K_{\gamma_I})$. Since $\Lambda$ is a directed set there exists $\lambda=(\hat{F},K_\gamma)\in \Lambda$ such that $\lambda\geq \lambda_I$ for every $I=1,\ldots,N$. Taking into account  the definition of the relation $\geq$ on $\Lambda$ we see that $\hat{F}$ contains all the operators $\{\hat{\varphi}_{\omega_1\omega'_1},\ldots,\hat{\varphi}_{\omega_N\omega'_N}\}$. Thus Condition \ref{f-Lambda} is satisfied.          

Condition \ref{RN} is ensured by Lemma \ref{Kg-R6N}.  Condition \ref{non-deg} is satisfied by virtue of the definition of the set $\Lambda$ presented in Section \ref{u-prop}.

Regarding Condition \ref{Q'=Q} we note first that, given $K_{\gamma},K_{\gamma'}$, there exists $K_{\gamma''}$ such that each d.o.f. in $K_{\gamma}\cup K_{\gamma'}$ is a linear combination of d.o.f. in $K_{\gamma''}$ (see Lemma \ref{K-linKg}). Suppose that $Q_{K_\gamma}=Q_{K_{\gamma'}}$. Then Equation \eqref{QKg-R6N} applied to $K_{\gamma''}$ allows us to use the following proposition \cite{q-nonl}: 

\begin{pro}
Let $K,K'$ be sets of independent d.o.f. of $N$ and $N'$ elements respectively such that $Q_K=Q_{K'}$. Suppose that there exists a set $\bar{K}$ of independent d.o.f. of $\bar{N}$ elements such that the image of $\tilde{\bar{K}}$ is $\R^{\bar{N}}$ and  each d.o.f. in $K\cup K'$ is a linear combination of d.o.f. in $\bar{K}$. Then each d.o.f. in $K$ is a linear combination of d.o.f. in $K'$.  
\end{pro}
\noindent Thus each d.o.f in $K_\gamma$ is a linear combination of d.o.f. in $K_{\gamma'}$. Then, as stated by Lemma \ref{lm-lin-comb}, $\gamma'\geq \gamma$ and Condition \ref{Q'=Q} follows.  

By virtue of the definition of the relation $\geq$ on $\Lambda$ (see Section \ref{u-prop}) $K_{\gamma'}\geq K_\gamma$ if $\gamma'\geq\gamma$. Thus Condition \ref{lin-comb} follows from Lemma \ref{lm-lin-comb}. Condition \ref{FF'} is satisfied  due to the same definition. 

We conclude that the directed set $(\Lambda,\geq)$ constructed in the present section for a theory of symmetric tensor fields of type $\binom{0}{2}$  satisfies all Conditions listed in Section \ref{ffH-sec} which means that for such a theory there exists the corresponding space $\S$ of projective quantum states. Let us emphasize that the directed set is built in a background independent manner hence the same can be said about the resulting space $\S$.     

\subsection{Projective quantum states for any tensor field theory}

A space of projective quantum states for any tensor field theory can be constructed in a fully analogous way to the space $\S$ built for a theory of symmetric tensor fields of type $\binom{0}{2}$ described in the previous section. The only things which have to be changed are elementary d.o.f. and, consequently, the definition of a set $K_\gamma$ of d.o.f. associated with a discrete frame $\gamma$ (definition of which remains unchanged).     

Let us consider a tensor field theory of a phase space $P\times Q$. A point $q$ of the configuration space $Q$ is a finite collection $(q^A)$, $A=1,2,\ldots,k$, of tensor fields  defined on a manifold $\Sigma$ such that $q^A$ is a tensor field of type $\binom{m_A}{n_A}$.  Then a point $p$ in the momentum space $P$ of the theory is a collection $(p_A)$ of tensor densities  of weight $1$ on $\Sigma$ such that the momentum $p_A$ conjugate to $q^A$ is a tensor density of type $\binom{n_A}{m_A}$. Given $A$, allowed tensor fields $q^A$ and $p_A$ may be subjected to some (consistent) symmetricity/antysymmetricity conditions.         

Suppose that either $m_A\neq 0$ or $n_A\neq 0$ and denote by $Y_A$ an ordered set of $m_A$ elements of $T_y\Sigma$ and $n_A$ elements of $T^*_y\Sigma$. Then the field $q^A$ can be evaluated on the set $Y_A$ yielding a real number. If $m_A=0=n_A$, that is, if $q^A$ is a function on $\Sigma$ then we denote by $Y_A$ a point $y\in \Sigma$  and evaluate the function $q^A$ at the point obtaining a real number.  In this way the set $Y_A$ defines a real function $\kappa_{Y_A}$ on $Q$. The following set
\[
{\cal K}:=\{\ \kappa_{Y_A}\ | \ Y_A\in (T_y\Sigma)^{m_A}\times(T^*_y\Sigma)^{n_A}, \ y\in\Sigma,\ A=1,2,\ldots, k \ \},
\]         
where $Y_A\in (T_y\Sigma)^{0}\times(T^*_y\Sigma)^{0}$ should be understood as $Y_A=y$, separates points in $Q$ and therefore can serve as a set of elementary configurational d.o.f..     

Assume again that either $m_A\neq 0$ or $n_A\neq 0$ and denote by $\omega^A$ an ordered set of $m_A$ vector fields on $\Sigma$ of compact support and of $n_A$ one-forms on the manifold of compact support. The momentum field $p_A$ contracted with elements of the set $\omega^A$ is a scalar density of weight $1$ and of compact support. This density can be integrated over $\Sigma$ which yields a real number. If $m_A=0=n_A$, that is, if $p_A$ is a scalar density on $\Sigma$ then we denote by $\omega^A$ a function on $\Sigma$ of a compact support. Then the density $\omega^A p_A$ (no summation over $A$ here) once integrated over $\Sigma$ yields a real number.  In this way the set $\omega^A$ defines a real function $\varphi_{\omega^A}$ on $P$.

The set $\cal F$ of all {\em non-zero} functions $\{\varphi_{\omega^A}\}_{A=1,2,\ldots,k}$ separates points in $P$ and can be chosen to be a set of elementary momentum d.o.f..  

Let $D=\dim\Sigma$. A discrete frame $\gamma$ is a  collection of vectors tangent to $\Sigma$   
\[
\gamma=\bigcup_{I=1}^N\{ \ e_{I1},e_{I2},\ldots,e_{ID}\ \}
\]  
such that 
\begin{enumerate}
\item for each $I\in\{1,\ldots, N\}$ the set $(e_{I1},e_{I2},,\ldots,e_{ID})$ is a basis of $T_{y_I}\Sigma$,
\item the points $\{y_1,\ldots,y_N\}$ are pairwise distinct.
\end{enumerate}
The set $\Gamma$ of all discrete frames is a directed set with a directing relation $\geq$ defined as before.

Each discrete frame $\gamma=\{e_{Ii}\}$ defines a set $K_\gamma$ of elementary d.o.f. in the following way. Let $y_I$ be a point underlying the frame $\gamma$ and let $(\theta^{I1},\theta^{I2},\ldots,\theta^{ID})$ be the dual basis to $(e_{I1},e_{I2},,\ldots,e_{ID})$. Given $A\in \{1,2,\ldots,k\}$, we define $Y_A$ using elements of both bases (if $q^A$ is a function then we set $Y_A=y_I$) and then $Y_A$ yields an elementary d.o.f. $\kappa_{Y_A}$. The set $K_\gamma$ is the set of all (pairwise distinct) configurational d.o.f. obtained according to this prescription.    

Carrying on further steps of the construction as it was done in the previous section we obtain a directed set $(\Lambda,\geq)$ for the tensor field theory under consideration which satisfies all Conditions listed in Section \ref{ffH-sec}. Therefore there exists a family of factorized Hilbert spaces labeled by elements of the directed set which provides us with a space $\S$ of projective quantum states for the theory. Clearly, this construction is also background independent.  

\section{Projective quantum states for theories of coupled fields}

\subsection{General construction}

Suppose that there are two field theories $T$ and $\bar{T}$ of phase spaces ${\cal V}$  and $\bar{\cal V}$ respectively and that we have coupled the fields of the theories obtaining thereby a new theory of a phase space ${\cal V}\times\bar{\cal V}$. Assume moreover that for both theories $T$ and $\bar{T}$ we have constructed spaces $\S$ and $\bar{\S}$ of projective quantum states. Can we use these spaces or some objects used to construct them to obtain projective quantum states for the theory of the coupled fields?

Taking into account the interpretation of the Kijowski's construction as a family of consistent experiments (see Section \ref{out}) it is reasonable to define the projective quantum states for the theory of the coupled fields using finite systems obtained from finite systems of $T$ and $\bar{T}$---if $\Lambda$ and $\bar{\Lambda}$ are directed sets of finite systems of, respectively, $T$ and $\bar{T}$, then we may choose a pair $(\lambda,\bar{\la})\in\Lambda\times\bar{\Lambda}$ to be a finite system for the new theory. Then a Hilbert space for this system would be $\h_{\la}\ot\h_{\bar{\la}}$. Can we then extend a family $\{\h_{\la}\ot\h_{\bar{\la}}\}$ to a  family of factorized Hilbert spaces?

The answer to this question is in affirmative. Let 
\begin{align*}
&\Big(\Lambda,\h_{\lambda},\tilde{\h}_{\lambda'\lambda},\Phi_{\lambda'\lambda},\Phi_{\lambda''\lambda'\lambda}\Big), &&\Big(\bar{\Lambda},\h_{\bar{\lambda}},\tilde{\h}_{\bar{\lambda}'\bar{\lambda}},\Phi_{\bar{\lambda}'\bar{\lambda}},\Phi_{\bar{\lambda}''\bar{\lambda}'\bar{\lambda}}\Big).
\end{align*}
be families of factorized Hilbert spaces used to construct the spaces $\S$ and $\bar{\S}$ respectively. Suppose now that $\Theta$ is a directed subset of the directed set $\Lambda\times\bar{\Lambda}$ and define
\begin{align}
  \h_\theta&:=\h_\la\ot\h_{\bar{\la}}, & \tilde{\h}_{\theta'\theta}&:=\tilde{\h}_{\la'\la}\ot\tilde{\h}_{\bar{\la}'\bar{\la}},\nonumber\\
  \Phi_{\theta'\theta}&:=F\circ(\Phi_{\la'\la}\ot\Phi_{\bar{\la}'\bar{\la}}), & \Phi_{\theta''\theta'\theta}&:=F\circ(\Phi_{\la''\la'\la}\ot\Phi_{\bar{\la}''\bar{\la}'\bar{\la}}),\label{phi-theta}
\end{align}
where $\theta=(\la,\bar{\la})$ (and analogously for $\theta'$ and $\theta''$), and
\[
F:\h_1\ot\h_2\ot\h_3\ot\h_4\to \h_1\ot\h_3\ot\h_2\ot\h_4
\]
is a ``flip isomorphism'' defined on simple elements as follows:
\[
F(v_1\ot v_2\ot v_3\ot v_4):= v_1\ot v_3\ot v_2\ot v_4.
\]
Then
\begin{equation}
\Big(\Theta,\h_{\theta},\tilde{\h}_{\theta'\theta},\Phi_{\theta'\theta},\Phi_{\theta''\theta'\theta}\Big)
\label{ffHs-th}
\end{equation}
is a family of factorized Hilbert spaces which provides us with a space of quantum states for the theory of the coupled fields.

To prove that \eqref{ffHs-th} is a family of factorized Hilbert spaces it is enough to show that the maps \eqref{phi-theta} form a commutative diagram analogous to \eqref{diagram}. The commutativity of the diagram can be expressed in the following form
\[
\Phi^{-1}_{\theta''\theta}\circ(\Phi^{-1}_{\theta''\theta'\theta}\ot\id)=\Phi^{-1}_{\theta''\theta'}\circ(\id\ot\Phi^{-1}_{\theta'\theta}).
\]
It is a simple exercise to show that this equality holds for every simple element of the tensor product
\[
\tilde{\h}_{\la''\la'}\ot\tilde{\h}_{\bar{\la}''\bar{\la}'}\ot\tilde{\h}_{\la'\la}\ot\tilde{\h}_{\bar{\la}'\bar{\la}}\ot\h_\la\ot \h_{\bar{\la}}.
\]

The result just obtained means that given two fields theories with corresponding families of factorized Hilbert spaces there are in general many distinct spaces of projective quantum states for a theory of the coupled fields which differ from each other by the choice of the directed set $\Theta$ of finite systems. It may seem that $\Theta=\Lambda\times\bar{\Lambda}$ is a natural choice, but we will argue in the next section that it is not always the case. On the other hand, the set $\Theta$ cannot be ``too small'' because then the resulting space $\S$ may be devoid of some relevant quantum d.o.f.. Therefore it seems safe to require that $\Theta$ is a cofinal directed subset\footnote{A subset $\Lambda'$ of a directed set $\Lambda$ is {\em cofinal} if for every $\lambda\in\Lambda$ there exists $\lambda'\in\Lambda'$ such that $\la'\geq\la$. A cofinal subset $\Lambda'$ of a directed set $\Lambda$ is naturally a directed set with the relation $\geq$ induced by that defined on $\Lambda$.}  of $\Lambda\times\bar{\Lambda}$. This requirement reduces totally the diversity of spaces of projective quantum states for the theory of the coupled fields since  for every two distinct cofinal subsets of $\Lambda\times\bar{\Lambda}$ the resulting spaces coincide---this fact follows directly from general properties of projective limits \cite{bour-sets}.           

\subsection{LQG coupled to tensor fields}

Consider a theory $T$ of some tensor fields defined on a four dimensional manifold $\cal M$. Assuming that ${\cal M}=\R\times\Sigma$, where $\Sigma$ is a three dimensional manifold and treating $\R$ as a ``time-axis'' one may cast the theory into Hamiltonian form. Then a point of a phase space of the theory consists of appropriate fields defined on the manifold $\Sigma$. Let 
\begin{equation}
\Big(\Lambda,\h_{\lambda},\tilde{\h}_{\lambda'\lambda},\Phi_{\lambda'\lambda},\Phi_{\lambda''\lambda'\lambda}\Big)
\label{ffHs-T}
\end{equation}
be a family of factorized Hilbert spaces built over the phase space of the theory $T$ according to the prescription presented in Section \ref{sec-tf}.

On the other hand one may treat the same manifold $\cal M$ as a space-time of General Relativity (GR) by equipping it with a Lorentzian metric $g$ subjected to the vacuum Einstein equations (the manifold can be already equipped with such a metric if it was applied to define the dynamics of $T$). A phase space of GR described in terms of the real Ashtekar-Barbero variables \cite{barb} (being fields on the same manifold $\Sigma$) underlies the construction of LQG. At the same time this phase space is the point of departure for the construction of a family 
\begin{equation}
\Big(\bar{\Lambda},\h_{\bar{\lambda}},\tilde{\h}_{\bar{\lambda}'\bar{\lambda}},\Phi_{\bar{\lambda}'\bar{\lambda}},\Phi_{\bar{\lambda}''\bar{\lambda}'\bar{\lambda}}\Big) 
\label{ffHs-lqg}
\end{equation}
of factorized Hilbert spaces described in \cite{proj-lqg-I} which gives the space $\D_{\rm LQG}$ of projective quantum states for (vacuum) LQG.

Assume now that we have coupled in a way GR with the theory $T$ and that we would like to obtain by a suitable quantization of this new theory a model of LQG coupled to canonical variables of $T$. Then as a space of quantum states for this new quantum model we may use a space obtained from the families \eqref{ffHs-T} and \eqref{ffHs-lqg} by a suitable choice of a directed set $\Theta\subset\Lambda\times\bar{\Lambda}$ of finite systems as described above. The only question we have to answer is how to choose the set $\Theta$?  

To this end let us describe briefly the set $\bar{\Lambda}$ introduced in \cite{proj-lqg-I}. The precise definition of $\bar{\Lambda}$ is complicated but we will not need here all those details. For our purpose it is enough to know that $\bar{\Lambda}$ is a cofinal directed subset of a directed set $\Lambda_{\rm Gra}\times \Lambda_{\rm Sfc}$, where $\Lambda_{\rm Gra}$ is the directed set of (finite) graphs in $\Sigma$ commonly used in LQG (see e.g. \cite{cq-diff}) and $\Lambda_{\rm Sfc}$ is a directed set  elements of which are finite collections of surfaces in $\Sigma$.       

Let us argue now that in the case of LQG and the theory $T$ the choice $\Theta=\Lambda\times\bar{\Lambda}$ is rather not a good one. Suppose then that an element $\lambda=(\hat{F},K_\gamma)\in\Lambda$ and an element $\bar{\lambda}=(\bar{\gamma},\sigma)\in \bar{\Lambda}$ are chosen in such a way that
\begin{enumerate}
\item the set $u$ of points underlying the discrete frame $\gamma$ has an empty intersection with every surface belonging to $\sigma$ and with the graph $\bar{\gamma}$;    
\item the supports of all fields $\{\omega^A\}$, which define operators  constituting a basis of $\hat{F}$ have empty intersections with every surface belonging to $\sigma$ and with the graph $\bar{\gamma}$.    
\end{enumerate}   
This means that $\lambda$ and $\bar{\lambda}$ are supported on disjoint subsets of $\Sigma$ and therefore quantum d.o.f. associated with $\lambda$ and $\bar{\lambda}$ cannot be coupled to each other. Thus it seems reasonable to not include elements $(\lambda,\bar{\lambda})$ of this sort to $\Theta$.     

Taking into account that every graph $\bar{\gamma}$ distinguishes a finite subset of $\Sigma$ consisting of all vertices of the graph it is natural to define the set $\Theta$ as follows: a pair 
\[
\big(\lambda=(\hat{F},\K_\gamma),\bar{\lambda}=(\bar{\gamma},\sigma)\big)\in \Lambda\times\bar{\Lambda}
\]
is an element of $\Theta$ if the set $u$ of points underlying the discrete frame $\gamma$ coincides with the set of all vertices of the graph $\bar{\gamma}$.      
\begin{lm}
$\Theta$ is a directed set.
\end{lm}

\begin{proof}
To prove the lemma it is enough to show that $\Theta$ is a cofinal subset of the directed set $\Lambda\times\bar{\Lambda}$. 

Consider then an arbitrary $\lambda=(\hat{F},K_\gamma)\in\Lambda$ and an arbitrary $\bar{\lambda}=(\bar{\gamma},\sigma)\in\bar{\Lambda}$ and denote by $u$ the set of points underlying the frame $\gamma$. Let $\bar{\gamma}'$ be a graph such that $u$ is a {\em proper} subset of the set of all vertices of the graph. Since $\bar{\Lambda}$ is a cofinal subset of $\Lambda_{\rm Gra}\times\Lambda_{\rm Sfc}$ there exists an element $\bar{\lambda}''=(\bar{\gamma}'',\sigma'')\in\bar{\Lambda}$ such that $\bar{\lambda}''\geq\bar{\lambda}$ and $\bar{\gamma}''\geq\bar{\gamma}'$. Therefore $u$ is a proper subset of the set of all vertices of $\bar{\gamma}''$.

Let $\gamma'$ be a discrete frame such that  the set $u'$ of points underlying the frame coincides with the set of all vertices of $\bar{\gamma}''$. Then $u$ is a proper subset of $u'$ and consequently $\gamma'\geq\gamma$. 

Let us recall that proving Proposition \ref{Lambda-dir} we considered a space $\hat{F}_0$ and a set $K_{\gamma''}$ of independent d.o.f. such that operators constituting a basis of $\hat{F}_0$ are linearly independent on $K_{\gamma''}$ and the number of the operators is lower than the number of elements of $K_{\gamma''}$. Then we showed that $\hat{F}_0$ can be enlarged to a space $\hat{F}''$ such that $(\hat{F}'',K_{\gamma''})$ belongs to the directed set $\Lambda$ considered in the proposition.

In the same way the space $\hat{F}$ considered here can be enlarged to a space $\hat{F}'$ such that the pair $\lambda':=(\hat{F}',K_{\gamma'})$ belongs to the set $\Lambda$ of finite systems defined for the theory $T$---because $u$ is a proper subset of $u'$ operators constituting a basis of $\hat{F}$ are linearly independent when restricted to $K_{\gamma'}$ and the number of the operators is lower than the number of elements of $K_{\gamma'}$.

Consequently, the pair $(\lambda',\bar{\lambda}'')$ belongs to $\Theta$ and
\[
(\lambda',\bar{\lambda}'')\geq(\lambda,\bar{\lambda})
\]  
which means that $\Theta$ is a cofinal subset of $\Lambda\times\bar{\Lambda}'$.  
\end{proof}

\section{Summary}

In this paper we constructed a space $\S$  of projective quantum states for any tensor field theory. Let us emphasize that this construction is very natural since it applies essential features of canonical variables of such a theory---in this case ''position'' variables are tensor fields and the configurational elementary d.o.f. are defined by evaluating the fields at vectors and covectors (co)tangent to a point, conjugate momenta are tensor densities and momentum elementary d.o.f. are defined as integrals of scalar densities obtained by contracting the momenta with vector fields and one-forms. Thanks to this choice of elementary d.o.f. the space $\S$ is built in a background independent manner. 

Although this space $\S$ can be used for a quantization of a tensor field theory our main goal was to use it for a construction of a space of projective quantum states for LQG coupled to tensor fields.

To this end we considered two theories for which spaces of projective quantum states are known and we showed how this knowledge can be used  to construct a space of projective quantum states for a theory being the result of a coupling of the two original theories. Next, applying this general construction and the space $\D_{\rm LQG}$ of projective quantum states for LQG introduced by Lan\'ery and Thiemann in \cite{proj-lqg-I} we obtained a space of such states for LQG coupled to tensor fields.  

Let us emphasize again that all constructions here are {\em kinematic} in this sense that they do not take into account dynamics of the theories under consideration and  constraints on the phase spaces.

The space of projective quantum states for a tensor field theory was constructed here on the basis of a general method introduced in \cite{q-nonl} and slightly modified in \cite{mod-proj}.  We would like to stress that this paper makes also a contribution to this general method---in Section \ref{u-prop} we described a fairly general scheme for constructing a directed set of finite physical systems from finite sets of configurational d.o.f. and finite dimensional spaces of momentum operators.


\end{document}